\newif\ifabstract
\abstracttrue
 \abstractfalse 
\newif\iffull
\ifabstract \fullfalse \else \fulltrue \fi

\documentclass[11pt]{article}

\usepackage{tikz}
\usetikzlibrary{tikzmark}
\newcommand{\tikzarc}[1]{%
\tikzmarknode{a}{#1}
\begin{tikzpicture}[overlay,remember picture]
\draw ([yshift=1pt]a.north west) to[bend left=20] ([yshift=1pt]a.north east);
\end{tikzpicture}%
}

\usepackage{amsfonts}
\usepackage{amssymb}
\usepackage{amstext}
\usepackage{amsmath}
\usepackage{xspace}
\usepackage{theorem}
\usepackage{graphicx}
\usepackage{url}
\usepackage{graphics}
\usepackage{colordvi}
\usepackage{colordvi}
\usepackage{subfigure}
\usepackage[noend]{algpseudocode}

\usepackage{graphicx,amssymb,amsmath}
\usepackage[ruled,vlined,linesnumbered]{algorithm2e}
\usepackage{xcolor}

\advance \oddsidemargin by -0.8in
\newcommand{\myparskip}{3pt}
\parskip \myparskip

        {\hspace*{\fill}$\Box$\par\vspace{4mm}}

\newcommand{\be}{\begin{enumerate}}
\newcommand{\ee}{\end{enumerate}}
\newcommand{\bd}{\begin{description}}
\newcommand{\ed}{\end{description}}
\newcommand{\bi}{\begin{itemize}}
\newcommand{\ei}{\end{itemize}}

\newtheorem{theorem}{Theorem}[section]
\newtheorem{lemma}[theorem]{Lemma}
\newtheorem{observation}[theorem]{Observation}

\newtheorem{definition}{Definition}[section]
\newtheorem{remark}{Remark}[section]
\newenvironment{proof}{\par \smallskip{\bf Proof:}}{\hfill\stopproof}
\def\stopproof{\square}
\def\square{\vbox{\hrule height.2pt\hbox{\vrule width.2pt height5pt \kern5pt
\vrule width.2pt} \hrule height.2pt}}

\usepackage{graphicx,amssymb,amsmath}
\usepackage[ruled,vlined]{algorithm2e}
\usepackage{xcolor}

\renewcommand{\phi}{\varphi}

\mathchardef\hyphen="2D

\newtheorem{problem}{Problem}

\begin{document}

\title{Dispersing Facilities on Planar Segment and Circle Amidst Repulsion}
\author{Vishwanath R. Singireddy
\and   Manjanna Basappa\thanks{Corresponding author
}}

\date{%
Department of Computer Science \& Information Systems, \\BITS Pilani, Hyderabad Campus, Telangana 500078, India\\{\tt \{p20190420,manjanna\}@hyderabad.bits-pilani.ac.in}\\%
    \today
}


\begin{titlepage}
\maketitle

\thispagestyle{empty}

\begin{abstract}
In this paper we consider the problem of locating $k$ obnoxious facilities (congruent disks of maximum radius) amidst $n$ demand points (existing repulsive facility sites) ordered from left to right in the plane so that none of the existing facility sites are affected (no demand point falls in the interior of the disks). We study this problem in two restricted settings: (i) the obnoxious facilities are constrained to be centered on along a predetermined horizontal line segment $\overline{pq}$, and (ii) the obnoxious facilities are constrained to lie on the boundary arc of a predetermined disk $\cal C$. An $(1-\epsilon)$-approximation algorithm was given recently to solve the constrained problem in (i) in time $O((n+k)\log{\frac{||pq||}{2(k-1)\epsilon}})$, where $\epsilon>0$ \cite{Sing2021}. Here, for the problem in (i), we first propose an exact polynomial-time algorithm based on a binary search on all candidate radii computed explicitly. This algorithm runs in $O((nk)^2\log{(nk)}+(n+k)\log{(nk)})$ time. We then show that using the parametric search technique of Megiddo \cite{MG1983}; we can solve the problem exactly in $O((n+k)^2)$ time, which is faster than the latter. Continuing further, using the improved parametric technique we give an $O(n\log^2 n)$-time algorithm for $k=2$. We finally show that the above $(1-\epsilon)$-approximation algorithm of \cite{Sing2021} can be easily adapted to solve the circular constrained problem of (ii) with an extra multiplicative factor of $n$ in the running time.

\end{abstract}

\end{titlepage}

\label{--------------------------------------------sec: intro---------------------------------------------------}
\section{Introduction}

The obnoxious facility location is a well-known topic in the operations research community. This paper addresses a variant of the problem, namely, the continuous obnoxious facility location on a line segment (\textsc{COFL}) problem, motivated by the following application. We wish to find locations for establishing $k$ obnoxious or undesirable facilities (such as garbage dump yards, industries generating pollution, etc.) along a straight highway road such that the pairwise distance between these new facilities and the distance between each of the new facilities and other existing non-obnoxious facilities (such as hospitals, schools, etc.) are maximized. The formal definition of the problem is below:

\vspace*{1mm}
\noindent The \textsc{COFL} problem: Given a horizontal line segment $\overline{pq}$ and an ordered set $P$ of $n$ points lying above a line through $\overline{pq}$ in the Euclidean plane, we aim to locate points on $\overline{pq}$ for centering $k$ non-overlapping congruent disks of maximum radius $r_{max}$ such that none of the points in $P$ lie interior to any of these disks.
\vspace*{1mm}

\noindent Recently, it has been shown in \cite{Sing2021} that we can solve the decision version of this problem (in which we are also given a radius $L$ as input along with $\overline{pq}$ and the set $P$ of $n$ points) in time $O(n+k)$ if the points in $P$ are given in order. Then, using this decision algorithm as a subroutine, an $(1-\epsilon)$-approximation algorithm ({\tt FPTAS}) has been given to solve the \textsc{COFL} problem in time $O((n+k)\log{\frac{||pq||}{2(k-1)\epsilon}})$, where $\epsilon>0$ and $||pq||$ is the length of the segment $\overline{pq}$ \cite{Sing2021}. 

In this work, we show that we can, in fact, solve the \textsc{COFL} problem exactly in polynomial time. Using the linear-time decision algorithm of \cite{Sing2021} again as a subroutine, we present two polynomial-time exact algorithms based on two different approaches: (i) the algorithm is based on doing a binary search on all candidate radii $L$ computed explicitly and runs in $O((nk)^2\log{(nk)}+(n+k)\log{(nk)})$ time, and (ii) the algorithm is based on Megiddo's parametric search \cite{MG1983}, and runs in $O((n+k)^2)$ time. We then discuss an $O(n\log^2 n)$-time algorithm to solve the \textsc{COFL} problem for $k=2$, which is faster than the previous two algorithms. Finally, we show how to solve the circular \textsc{COFL} problem using the algorithm of \cite{Sing2021}.

\section{Related Work}
Several variants of obnoxious facility location problems have been investigated by the operations research and computational geometry communities. A recent review of many variants of obnoxious facility location problems and specialized heuristics to solve them may be found in \cite{Chur2022}. Kartz et al. \cite{Katz2002} studied two variants of $k$-obnoxious facility location problems (i) max-min $k$-facility location, in which the distance between any two facilities is at least some fixed value and the minimum distance between a facility and a demand point is to be maximized, and (ii) min-sum one-facility location problem in which the sum of weights of demand points lying within a given distance to the facility is to be minimized. They solved the first problem in $O(n\log^2{n})$ time for $k=2$ or $3$ and for any $k\geq 4$ their algorithm runs in $O(n^{(k-2)}\log^2{n})$ time for rectilinear case. In \cite{Katz2002}, they also solved the Euclidean case where their algorithm runs in $O(n\log{n})$ time for $k=2$. For the second problem, their algorithm runs in $O(n\log{n})$ time (rectilinear case) and $O(n^2)$ time (Euclidean case). Qin et al. \cite{Qin2000} studied a variant of  $k$-obnoxious facility location problem in which the facilities are restricted to lie within a given convex polygonal domain. They proposed a 2-factor approximation algorithm based on a Voronoi diagram for this problem, and its running time is $O(kN\log{N})$, where $N=n+m+k$, $n$ denotes the number of demand points, $m$ denotes the number of vertices of the polygonal domain and $k$ is the number of obnoxious facilities to be placed. D{\'\i}az-B{\'a}{\~n}ez et al. \cite{Diaz2003} modelled obnoxious facility as empty circular annulus whose width is maximum. This refers to a max-min facility location problem such that the facility is a circular ring of maximum width, where the width is the absolute difference between the radii of the two concentric circles. They solved the problem in $O(n^3\log{n})$ time, and if the inner circle contains a fixed number of points then the problem can be solved in $O(n\log{n})$ time. Maximum width empty annulus problems for axis parallel square and rectangles can be solved in $O(n^3)$ and $O(n^2\log{n})$ time respectively \cite{Bae2021}. Abravaya and Segal \cite{Abra2010} studied the problem of locating maximum cardinality set of obnoxious facilities within a bounded rectangle in the Euclidean plane such that their pairwise distance is at least a given threshold. They proposed a 2-approximation algorithm and also a PTAS based on shifting strategy \cite{Hoch1985}. Agarwal et al. \cite{Agar1994} showed the application of Megiddo’s parametric search \cite{MG1983} method to solve several geometric optimization problems in the Euclidean plane such as the big stick problem, the minimum width annulus problem, and the problem of finding largest mutual visible spheres.

\section{Polynomial time exact algorithms}
Given a horizontal segment $\overline{pq}$, without loss of generality, we can assume that all the points in $P$ lie in the upper half-plane defined by a line through $\overline{pq}$. After fixing the segment $\overline{pq}$, we define the following decision version of the \textsc{COFL} problem as follows:

 \begin{enumerate}
 \item[$ $]\textsc{Dcofl}$(P, k, L)$: Given an ordered set $P$ of $n$ points lying above a line through $\overline{pq}$, a real number $L$ and an integer $k$, can we pack $k$ non-overlapping disks of radius $L$ centered on $\overline{pq}$ such that none of these disks contains a point of $P$ in their interior?
  \end{enumerate}
\begin{algorithm}[tb]
  \caption{Greedy\_LPacking$(P, k, L)$}\label{alg1}
	
\SetAlgoLined	  
 Compute $I$ \\
  Compute $I^\mathsf{c}$ where the elements are ordered from $p$ to $q$, and let $m=|I^\mathsf{c}|$ \\
	   $j\leftarrow 0$ \\
	  		\For{each $i\leftarrow 1$ to $m$}{
				$\gamma \leftarrow \bigr\lfloor \frac{\text{length of } i \text{th interval of }I^\mathsf{c}}{2L} \bigr\rfloor$ \\
				\If{$(j+\gamma+1)\leq k$}{
				     On the $i$th interval of  $I^\mathsf{c}$, pack the disks $d_{j+1},d_{j+2},\dots,d_{j+\gamma+1}$ of radius $L$ \\ 
				     update $j\leftarrow (j+\gamma+1)$\\
				     update $I^\mathsf{c}$ such that the distance between the left end point $r_i$ of the left most interval $[r_i,l_{i+1}]$ in $I^\mathsf{c}$ and the right most point of $\partial d_j$ on $\overline{pq}$ is at least $2L$.\\
				     \If{$(j=k)$}{\textbf{break}}}
				\Else{On the $i$th interval of  $I^\mathsf{c}$, pack the disks $d_{j+1},d_{j+2},\dots,d_{k}$ of radius $L$ \\ 
			 update $j\leftarrow k$\\
			 \textbf{break}}}

\If{ $(j=k)$}{
         \Return{$(\textsc{yes},\ \{ d_1, d_2, \ldots, d_k\})$}}
         \Else{ 
         \Return{$(\textsc{no},\ \emptyset)$}}
      
\end{algorithm}	
  
  Let $I=\{ [x_{i,1}, x_{i,2}]| p_i\in P\}$ be the set of intervals on $\overline{pq}$ ordered from $p$ to $q$, where every point in each interval $[x_{i,1}, x_{i,2}]$ is at distance at most $L$ from $p_i\in P$ (see Fig. \ref{figp}). The decision algorithm returns \textsc{yes} if we can place $k$ pairwise interior-disjoint disks of radius $L$, centered on $\overline{pq}$ but outiside the intervals in $I$. Otherwise the algorithm returns \textsc{no}. Let $I^\mathsf{c}=\{ [x_{i,2}, x_{j,1}]| p_i, p_j \in P, i<j\}\cup \{[x_{0,2}, x_{1,1}],[x_{n,2}, x_{n+1,1}]\}$ be the set of complemented intervals of $I$, where every point in each interval $[x_{i,2}, x_{j,1}]\in I^\mathsf{c}$ is at distance at least $L$ from every point in $P$ and $x_{0,2}=x(p)$, $x_{n+1,1}=x(q)$. We use the algorithm of \cite{Sing2021} for solving the decision version of the problem. For completeness and for the purpose of describing exact algorithms here, we have reproduced that decision algorithm again here (see Algorithm \ref{alg1}). Let ${\cal R}$ be the the Minkowski sum of $\overline{pq}$ and a disk of radius $L$, i.e., the union of $\overline{pq}$ and the region swept by a disk of radius $L$ when its center moves along $\overline{pq}$.
\begin{figure}[!ht]
\begin{centering}
\includegraphics[scale=0.6]{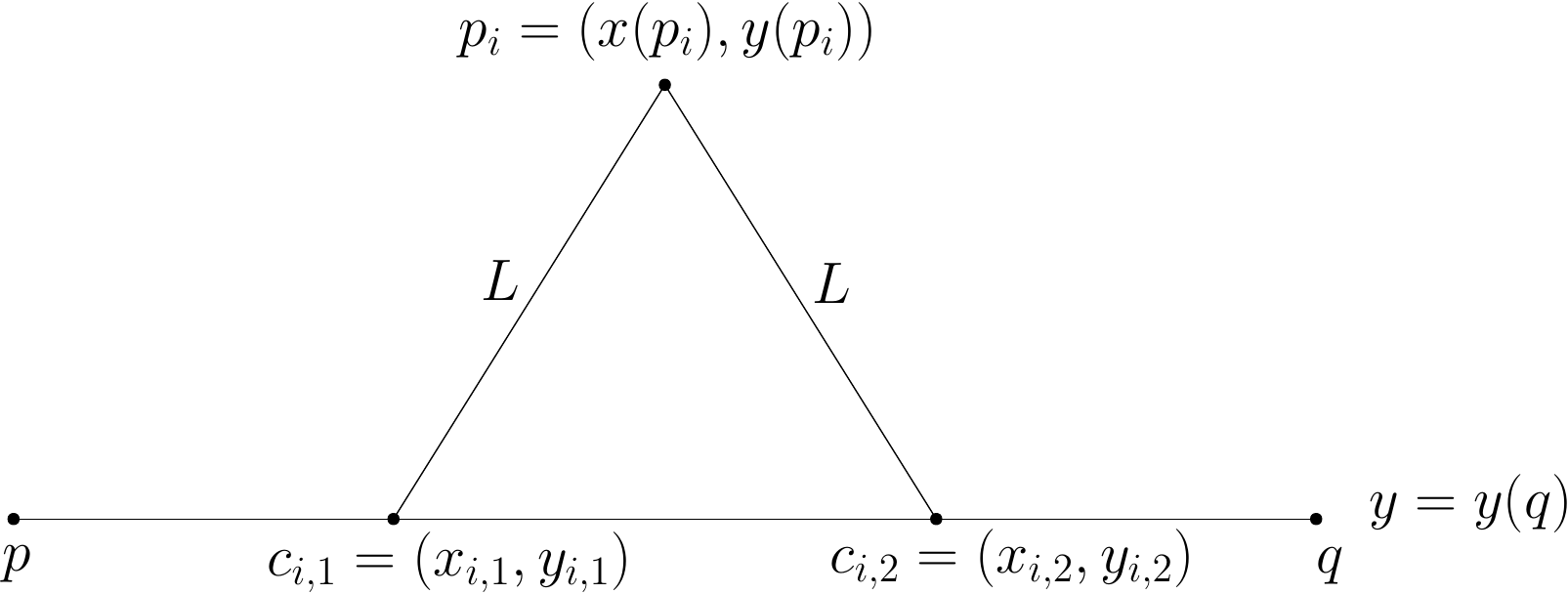}
\caption{A point $p_i$ is having two center points of $\overline{pq}$ which are at unknown distance $L$}
\label{figp}
\end{centering}
\end{figure}
  
\subsection{Algorithm based on binary search on all candidate radii values}
The approach that we follow here is that we first compute the set ${\cal L}_{\ell_2}$ (resp. ${\cal L}_{\ell_{\infty}}$) of all candidate radii $L$, based on some observations on optimal packing, where ${\cal L}_{\ell_{\infty}}$ and ${\cal L}_{\ell_2}$ are for rectilinear space and Euclidean space respectively. We sort ${\cal L}_{\ell_2}$ (resp. ${\cal L}_{\ell_{\infty}}$), then invoke Algorithm \ref{alg1} by setting $L$ to some element in ${\cal L}_{\ell_2}$ (resp. ${\cal L}_{\ell_{\infty}}$). Ultimately, we find $r_{max}$ by doing binary search on the ordered elements in ${\cal L}_{\ell_2}$ (resp. ${\cal L}_{\ell_{\infty}}$).

To get familiarity with this approach, we first discuss this for the rectilinear version of the \textsc{COFL} problem (in $\ell_1$ or $\ell_{\infty}$ metric). We then apply the approach to solve the (Euclidean) \textsc{COFL} problem (i.e., in $\ell_2$ metric). We now define the rectilinear version of the problem as follows:
\begin{problem}
 Given a horizontal line segment $\overline{pq}$, an ordered set $P$ of $n$ demand points, and an integer $k$, pack $k$ maximum-size axis-aligned squares centered on $\overline{pq}$ such that no point of $P$ lies in the interior of any of these squares, where the size of a square is defined to be half of its side length.
\end{problem}

The optimal size of the square is determined by searching all possible candidate sizes as discussed below:\\
Consider an instance of the \textsc{COFL} problem, i.e., a line segment $\overline{pq}$, an ordered set $P$ of $n$ points, and an integer $k$. As in the previous section, without loss of generality, let all the points of $P$ be lying above the line through the segment $\overline{pq}$. Now, we need to place $k$ non-overlapping squares of maximum size with centers lying on $\overline{pq}$ without violating the constraint that no point of $P$ lies inside any of these squares. The maximum size of the square depends on some "{\it influencing points}" in $P$ for the given position of the fixed horizontal line segment $\overline{pq}$. 
\begin{definition}
The influencing points ${P}^{inf}\subseteq P$ are those points which satisfy the following criteria: given any point $p_i\in {P}^{inf}$ (possibly along with another point $p_j\in {P}^{inf}$) we can place $k$ pairwise disjoint congruent squares centered on $\overline{pq}$ such that (i) $p_i$ (and $p_j$) lies on the boundary of some of these squares, and (ii) if $p_i$  (or $p_j$ or both) is removed from $P$, then the size of these squares may be increased while their centers are perturbed on $\overline{pq}$ to keep their pairwise disjointness intact, and (iii) no other points of $P$ lie in the interior of the squares all the time.
\end{definition}

\begin{lemma}\label{linfinity}
 If the position of $\overline{pq}$ is fixed and an ordered set $P$ of $n$ points are not lying strictly below the line through $\overline{pq}$, then for any optimal solution to the \textsc{COFL} problem in $\ell_{\infty}$ metric, there are at most two points in ${P}$ that will determine the optimal radius $r_{max}$ for the disks (axis-alinged squares with size $r_{max}$) in the packing, and also $|{\cal L}_{\ell_{\infty}}|=O(n^2k^2)$.
\end{lemma}
\begin{proof}
The proof follows from the following cases:
\begin{itemize}
\setlength{\itemindent}{1.5em}
\item[\textbf{Case 0:}] If ${P}^{inf}$ is empty, then no points of $P$ are at a distance of $\frac{||pq||}{2(k-1)}$ from the line segment $\overline{pq}$. This means that none of the points of $P$ lie on the boundary of the $k$ squares in the optimal packing. Therefore, we have that the maximum size of the square is $\frac{||pq||}{2(k-1)}$.
\end{itemize}
If ${P}^{inf}$ is non-empty, then there will be at least one point on the boundary of any of the $k$ squares in the optimal packing.

First, we will discuss how to find the set ${P}^{inf}\subseteq P$ of all the {\it influencing points}. Let $s_1, s_2, \ldots, s_k$ be the squares in the optimal solution. Let the point $p'\in P$ be the leftmost point among all the points of $P$ that lie on the boundaries of the squares $s_j$ (where $j=1,2,\ldots,k$) in the optimal packing (in case of ties, we may consider the one with largest $y$-coordinate). We will see that this point $p'\in {P}^{inf}$ in many cases, and sometimes $p'\notin {P}^{inf}$ (this means that $p'$ is not involved in determining $r_{max}$ and it may be that some other two points $p_i,p_{i'}\in {P}^{inf}$ on the right of $p'$ determines the radius $r_{max}$). Now, assume that $p'$ participates in determining $r_{max}$, let $p_i=p'$. Then, $p_i$ may lie on the right edge, or the left edge or the top edge boundaries of some square $s_j$. Based on these, we have the following cases.

\begin{itemize}
\setlength{\itemindent}{1.5em}
\item[\textbf{Case 1:}] Assume that $p_i$ lies on the right edge of $s_j$ such that its position determines the size of the square (see Fig. \ref{figob5}). For each $j=1,2, \ldots, k$ in the optimal packing, $p_i$ can lie on the right edge of $s_j$. For each such point $p_i\in {P}^{inf}$ the candidate radius $r_{\textsc{can}}=\frac{x(p)-x(p_i)}{2j-1}$ can be calculated and stored in ${\cal L}_{\ell_{\infty}}$. Hence, the number of candidate radii in this case is $O(nk)$.  

\begin{figure}[!htb]

\begin{centering}
     \includegraphics[width=.8\linewidth]{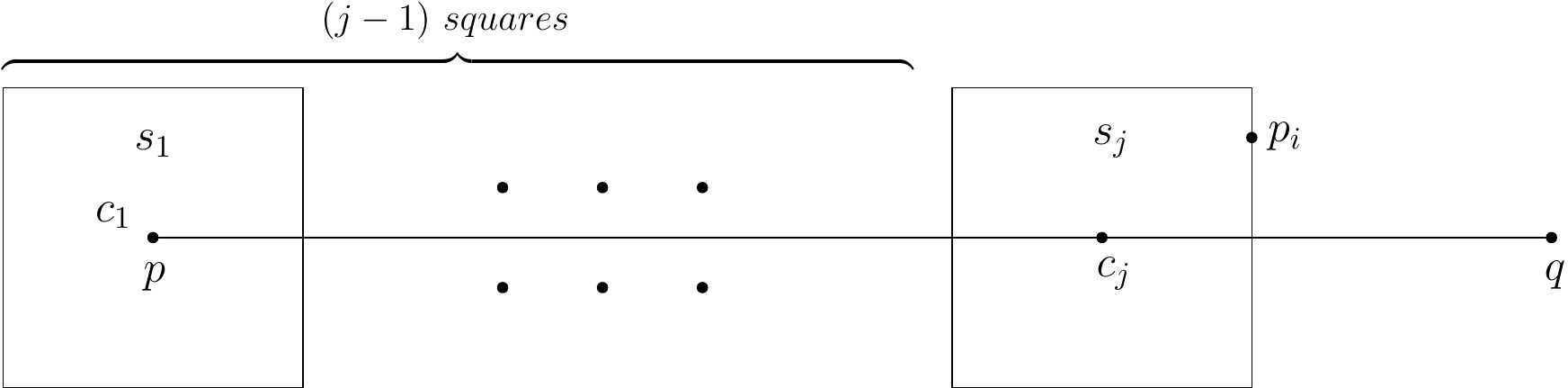}
     \caption{Illustration of case 1}\label{figob5}
     \end{centering} 
\end{figure}
\item[\textbf{Case 2:}] If $p_i$ lies on the top edge of $s_j$, then there are three subcases:
\begin{itemize}
\item[(i)] There is another point $p_{i'}\in P$ that lies on the right edge of $s_j$ (see Fig. \ref{figobj4c1}) or on the right edge of $s_{j'}$, where $j'\geq j+1$ (see Fig. \ref{figobj4c2}). Then if this pair $p_{i}$, $p_{i'}$ determines the maximum size of the square then the candidate radius $r_{\textsc{can}}=y(p_i)-y(q)\in {\cal L}_{\ell_{\infty}}$ for every pair $p_{i}, p_{i'}\in {P}^{inf}$. Note that this can be done exhaustively for every pair of points in $P$ and for very possible value of $j$ by checking that no other point of $P$ lies inside any of these $j$ squares.
 \begin{figure}[!htb] 
    \begin{centering}
     \includegraphics[width=.8\linewidth]{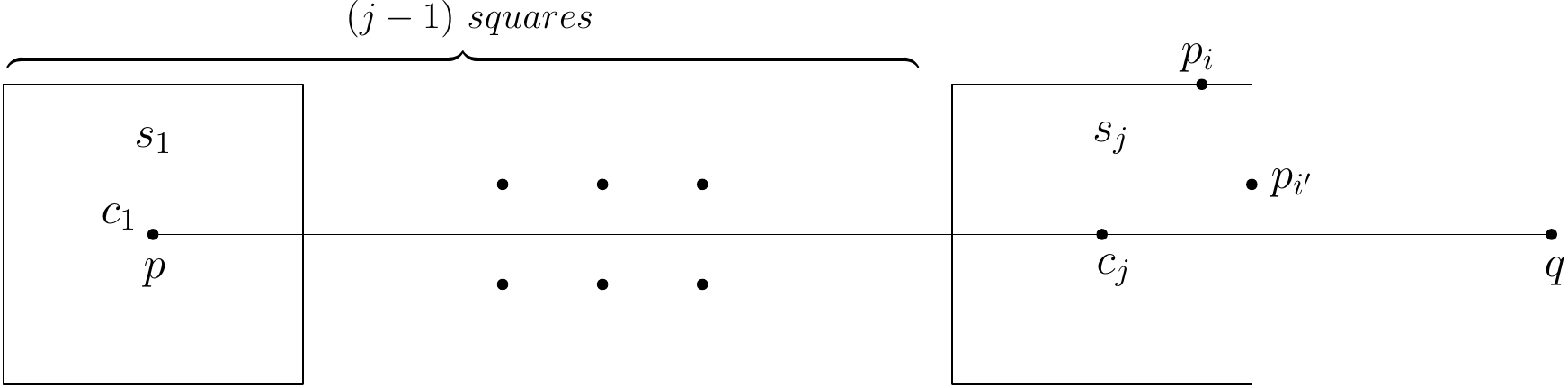}
     \caption{Illustration of case 2(i) with $j'=j$}\label{figobj4c1}
\end{centering} 
 \end{figure}
 
 \begin{figure}[!htb] 
    \begin{centering}
     \includegraphics[width=.8\linewidth]{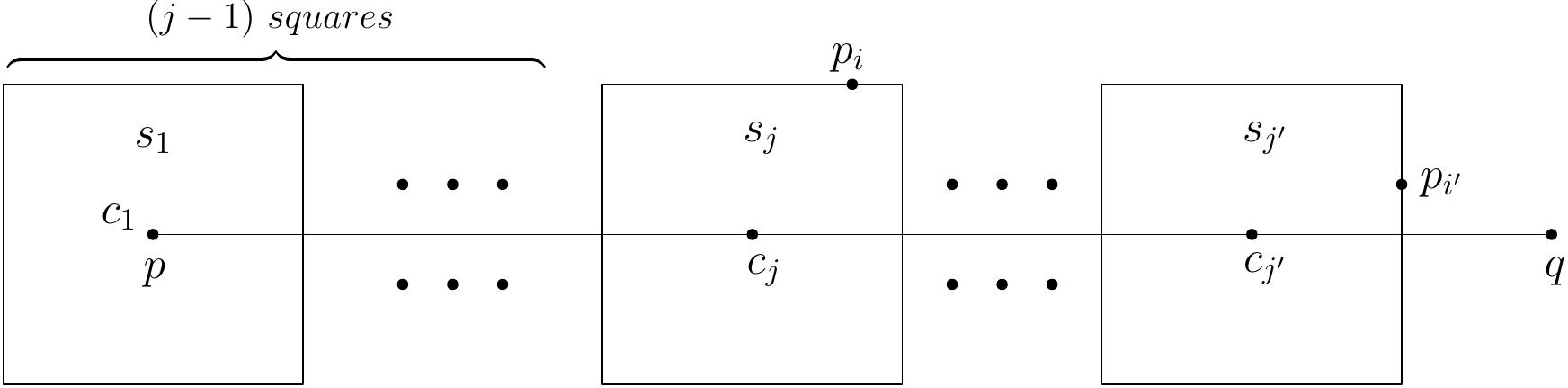}
     \caption{Illustration of case 2(i) with $j'\geq j+1$}\label{figobj4c2}
\end{centering} 
 
\end{figure}

\item[(ii)] If there is no point lying on the right edge of $s_{j'}$ for any $j'=j+1, j+2, \ldots, k$ (see Fig. \ref{figob4}) then this case will be same as \textbf{Case 0} if $( \bigcup_{j=1}^{k}s_j)\cap[p,q]=[p,q]$. Otherwise, $r_{\textsc{can}}=y(p_{i})-y(q)$, where $i \in \{1,2,\dots,n\}$ is the index of the point in $P$ with the smallest $y$-coordinate such that $r_{\textsc{can}}/2+x(p)\leq x(p_{i})\leq x(q)+r_{\textsc{can}}/2$. 
\end{itemize}

 \begin{figure}[!htb] 
    \begin{centering}
     \includegraphics[width=.8\linewidth]{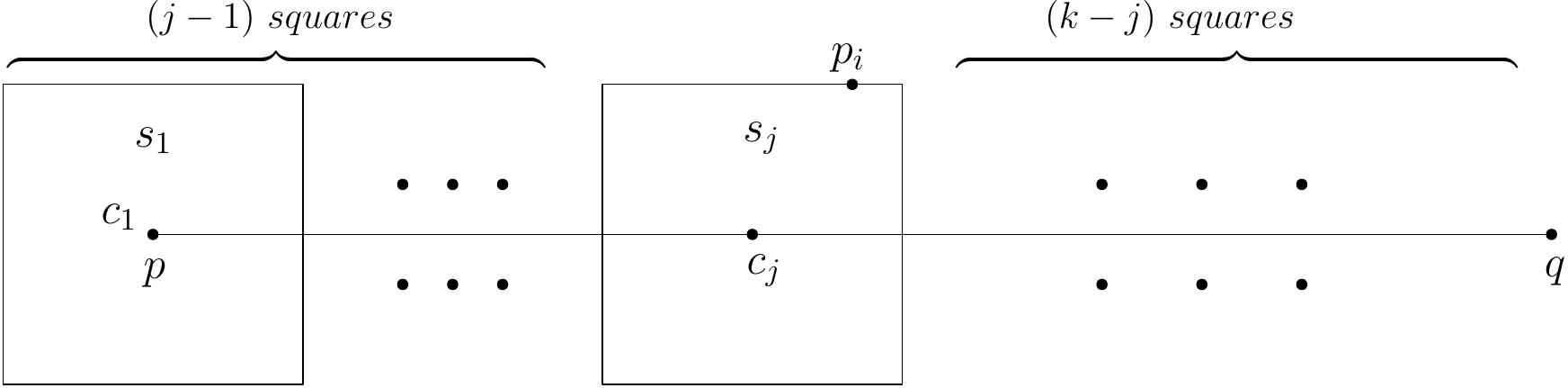}
     \caption{Illustration of case 2(ii) with no $p_{i'}$}\label{figob4}
\end{centering} 
 
\end{figure}

\item[\textbf{Case 3:}] Now assume that the point $p_i$ lies on the left edge of $s_j$ for $j=1,2, \ldots, k$ (see Fig. \ref{figobj6c1}). Observe that if $p_i$ lies on the left edge of $s_j$ where $j\geq 2$, then without loss of generality we can assume that all the first $j-1$ squares are compactly packed together and the right edge of $s_{j-1}$ coincides with the left edge of $s_j$. Thus, this is same as one of the above cases as the point $p_i$ lies on the right edge of $s_{j-1}$. Hence, in the remaining cases we can assume that $p_i$ always lies on the left edge of $s_1$.

 \begin{figure}[!htb] 
    \begin{centering}
     \includegraphics[width=.8\linewidth]{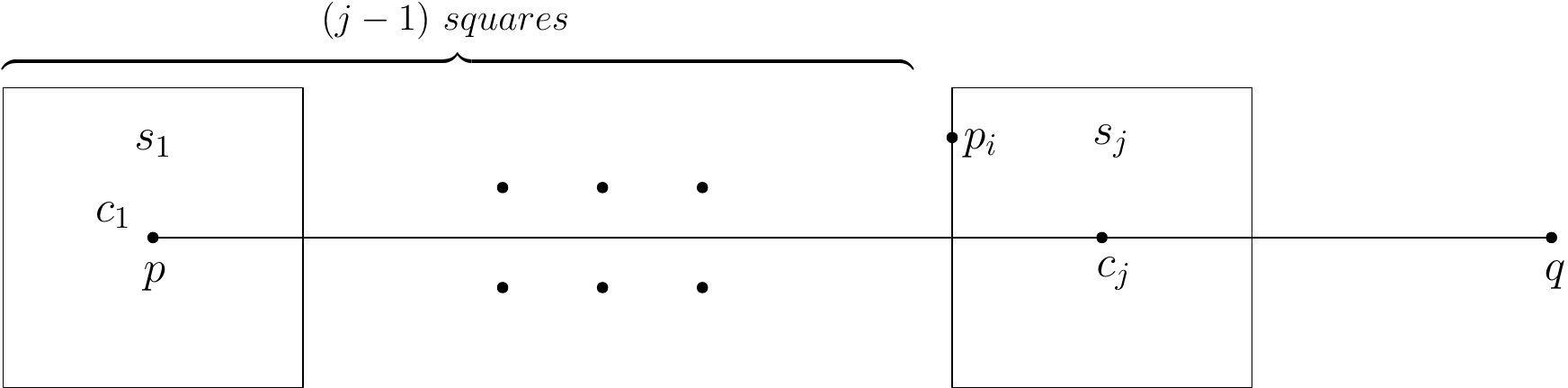}
     \caption{$p_i$ lies on the left edge of the square $s_j$}\label{figobj6c1}
\end{centering} 
 
\end{figure}
\begin{itemize}
 \item[\textbf{(i)}] There is no point $p_{i'}$ lying on either the top edge or right edge of any of the squares $s_{j'}$ for $j'=2,3, \ldots, k$. This is essentially the same case as \textbf{Case 0} except that the segment $\overline{pq}$  is shrunken to be of length $||pq||-(x(p_i)-x(p)+r_{\textsc{can}})$ and $s_1$ is centered at the left end point of this shrunken segment. Then, $r_{\textsc{can}}=\frac{||pq||-x(p_i)+x(p)}{2k-1}\in {\cal L}_{\ell_{\infty}}$.
 \item[\textbf{(ii)}] If $p_i$ lies on the left edge of the square $s_1$ then based on the position of the point $p_{i'}$ $(i'>i)$  that lies on the square $s_1$ or the square $s_{j'}$ $(j'>j=1)$ on the right side of $p_i$ we have the following situations.
 \begin{itemize} 
\item[(a)]When $p_{i'}$ lies on the right edge of $s_1$. Then, $r_{\textsc{can}} = \frac{x(p_{i'})-x(p_i)}{2} \in {\cal L}_{\ell_{\infty}}$.

 \begin{figure}[!htb] 
    \begin{centering}
     \includegraphics[width=.8\linewidth]{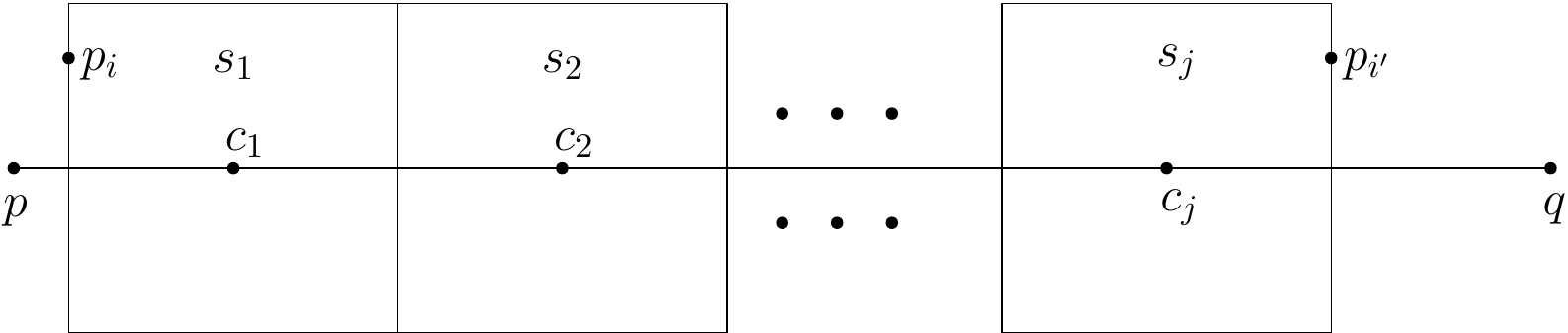}
     \caption{Illustration of case 3(ii)(b)}\label{figobj6c3}
\end{centering} 
 \end{figure}
\item[(b)] When $p_{i'}$ lies on the right edge of the square $s_{j'}$ (see Fig. \ref{figobj6c3}). Then, $r_{\textsc{can}}=\frac{x(p_{i'})-x(p_{i})}{2j-1}\in {\cal L}_{\ell_{\infty}}$, $j'=2,3,\dots,k$.

\item[(c)] When $p_{i'}$ lies on the top edge of the square $s_j$ (see Fig. \ref{figobj6c4}). Then, $r_{\textsc{can}}=y(p_{i'})-y(p)\in {\cal L}_{\ell_{\infty}}$. It is similar to \textbf{Case 2(i)}.
\end{itemize}
\item[\textbf{Case 4:}] Now assume any two distinct points $p_i, p_{i'} \in P$ being located such that they lie on the adjacent sides of the same square $s_j \ (j=j')$ (see Fig. \ref{figobj6c4f}). Then, if we consider the size of such a square for each pair of points in $P$, there are at most $O(n^2)$ candidate radii $r_{\textsc{can}}$ in  ${\cal L}_{\ell_{\infty}}$.
\end{itemize}

Observe that all the remaining cases can be reduced to one of the above cases. Since there are only a constant number $c$ of possible positions of the two squares $s_j,s_{j'}$ in the optimal packing for a fixed pair $p_i, p_{i'} \in P^{inf}$, the number of candidate radii $r_{\textsc{can}}$, i.e., $|{\cal L}_{\ell_{\infty}}|={n\choose 2}\sum\limits_{j=1}^k\sum\limits_{j'=j}^k c=O(n^2k^2)$. Hence the cardinality of ${\cal L}_{\ell_{\infty}}$ is $O((nk)^2)$, we can compute ${\cal L}_{\ell_{\infty}}$ in $O((nk)^2)$ worst case times
\end{itemize}
\end{proof}
   
   \begin{figure}[!htb] 	
    \begin{centering}
     \includegraphics[width=.8\linewidth]{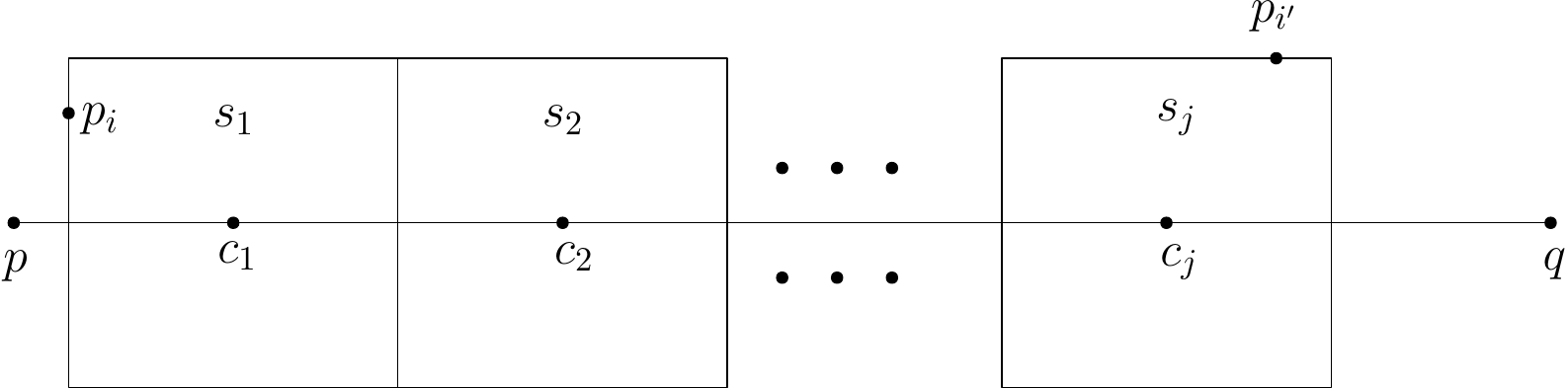}
     \caption{Illustration of case 3(ii)(c)}\label{figobj6c4}
\end{centering} 
\end{figure}
  \begin{figure}[!htb] 
    \begin{centering}
     \includegraphics[width=.8\linewidth]{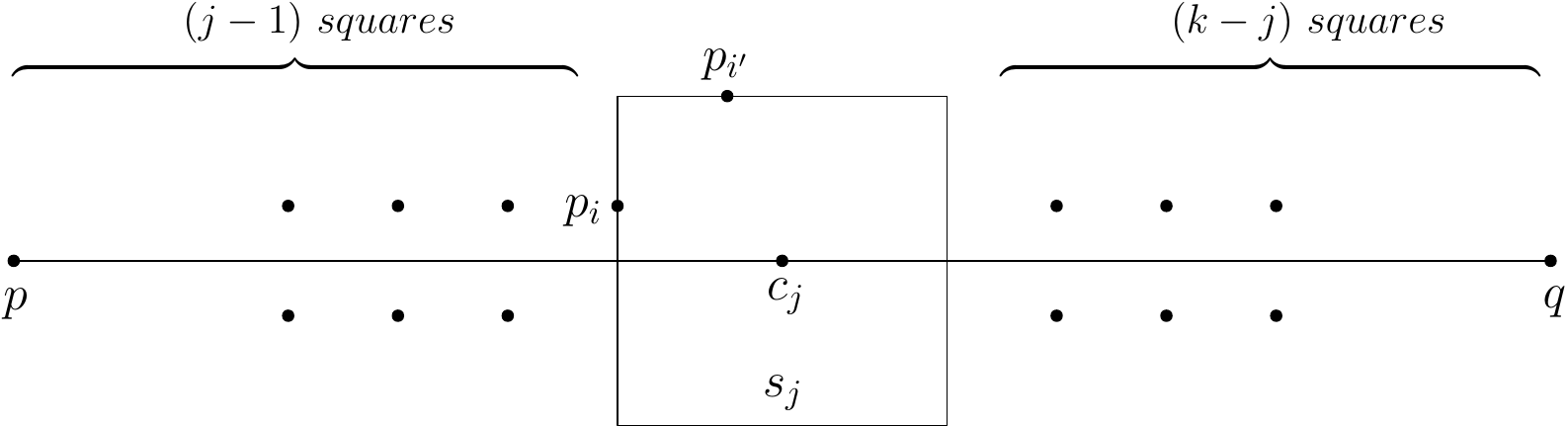}
     \caption{Illustration of case 4}\label{figobj6c4f}
\end{centering} 
 \end{figure}

Now, sort the elements in ${\cal L}_{\ell_{\infty}}$. By doing binary search on the sorted ${\cal L}_{\ell_{\infty}}$, solve the \textsc{Dcofl}$(P, k, L)$ problem repeatedly (invoke Algorithm \ref{alg1}) by setting $L$ to the median $r_{\textsc{can}} \in {\cal L}_{\ell_{\infty}}$ each time. Since the set of candidate radii computed is exhaustive, the above search guarantees to find $r_{\textsc{can}}= r_{max}$. Hence we have the following theorem.
\begin{theorem}\label{th3}
For a given line segment $\overline{pq}$ and an ordered set $P$ of $n$ points in the rectilinear plane, we can solve the \textsc{COFL} problem optimally in $O((nk)^2\log{(nk)}+(n+k)\log{(nk)})$ time.
\end{theorem}

\begin{remark}
  The {\it COFL} problem for $k=2$, under the $\ell_{\infty}$ norm, can be solved in $O(n\log n)$ by employing the optimization technique of Frederickson and Johnson \cite{GF1984} instead of doubling search and bisection methods, because the problem will be a special case of the one in \cite{Katz2002}.
\end{remark}

\subsubsection*{The Euclidean \textsc{COFL} problem:}
The approach here is that we first compute the set ${\cal L}$ of all candidate radii $L$, based on some observations of possible positions of all $k$ disks in an optimal packing. We sort all the elements in ${\cal L}$. Then, by doing a binary search on the sorted ${\cal L}$, we invoke the decision algorithm (Algorithm \ref{alg1}) each time we peek an element at the middle index by setting $L$ to this element as the candidate radius. We continue this search until we find the radius $L^*$ (maximum radius) such that for all $L\in {\cal L}$ with $L> L^*$ the decision algorithms returns \textsc{no}. Hence, $L^*=r_{max}$.

Consider an instance of the \textsc{COFL} problem, i.e., a line segment $\overline{pq}$, an ordered set $P$ of $n$ points, and an integer $k$. As already assumed, without loss of generality, let all the points of $P$ be lying above the line through the segment $\overline{pq}$. Now, we need to place $k$ non-overlapping congruent disks of maximum radius with centers lying on $\overline{pq}$ without violating the constraint that no point of $P$ lies inside any of these disks. The maximum radius of the disks depends on some "{\it influencing points}" in $P$ for the given position of the fixed horizontal line segment $\overline{pq}$. 

\begin{definition}
The influencing points ${P}^{inf}\subseteq P$ are those points which satisfy the following criteria: given any point $p_i\in {P}^{inf}$ (possibly along with another point $p_{i'}\in {P}^{inf}$) we can place $k$ pairwise disjoint congruent disks centered on $\overline{pq}$ such that (i) $p_i$ (and $p_{i'}$) lies on the boundary of some of these disks, and (ii) if $p_i$  (or $p_{i'}$ or both) is removed from $P$, then the radius of these disks may be increased while their centers are perturbed on $\overline{pq}$ to keep their pairwise disjointness intact, and (iii) none of the other points of $P$ lie in the interior of any of these disks before and after the radius increases.
\end{definition}

\begin{lemma}\label{leuclidean}
 Given that the position of $\overline{pq}$ is fixed, for any optimal solution to the \textsc{COFL} problem in $\ell_{2}$ metric, there are at most two points in ${P}^{inf}\subseteq P$ that will determine the optimal radius $r_{max}$ of the disks in the corresponding packing, and also $|{\cal L}|=O(n^2k^2)$.
\end{lemma}

\begin{proof} The proof is based on case analysis. First, let $r_{\textsc{can}}\in {\cal L}$ be a candidate radius.\\
\textbf{Case 1:} The set ${P}^{inf}$ is empty and then $r_{\textsc{can}}=\frac{||pq||}{2(k-1)}$.

If this case is not satisfied, then there will be at least one point lying on the boundary of some disk in the optimal packing which will influence on the value of $r_{max}$. Among these points in $P$ lying on the boundaries of the optimal disks, let $p_i$ be the left most, i.e., the point with the smallest $x$-coordinate. Let $d_1, d_2, \ldots, d_k$ be the disks ordered from left to right in the optimal packing. Now we will examine all possible positions of these disks for every subset of at most two influencing points in ${P}^{inf}$. We will also show how to compute the corresponding candidate radii $r_{\textsc{can}}\in {\cal L}$. To this end, consider a disk $d$ centered on $\overline{pq}$. Divide the part of the boundary arc $\partial d$ of $d$ lying above $\overline{pq}$ into left and right arc segments by a vertical line through the center of $d$ (see Figure \ref{figl2arc}). We then show that at most two points $p_i,p_{i'}\in {P}^{inf}\subseteq P$ determine the radius $r_{max}$ of the disks in an optimal packing.

\begin{figure}[!htb]   
   \begin{centering}
       \includegraphics[width=.8 \linewidth]{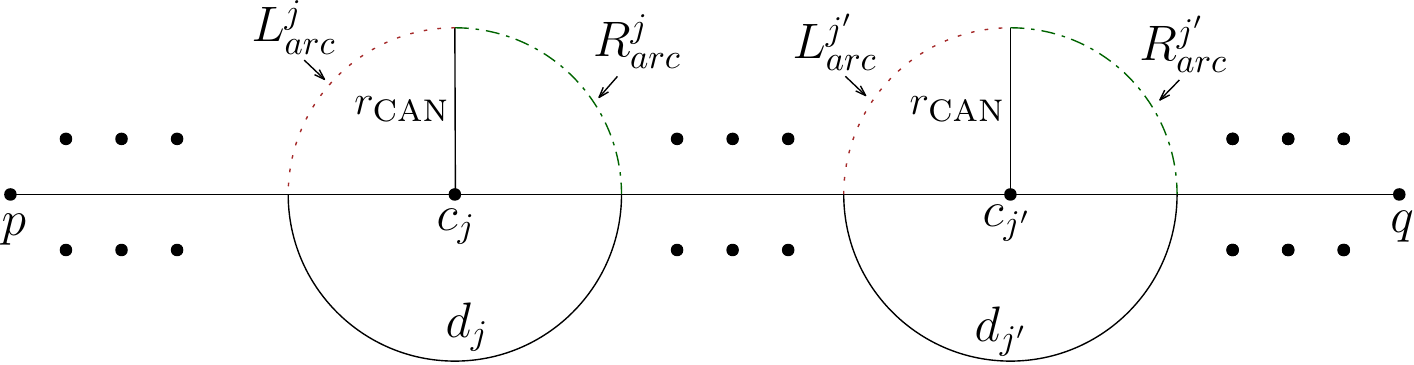}
     \caption{Left arc segment and right arc segment of disks $d_j$ and $d_{j'}$}\label{figl2arc}
\end{centering} 
\end{figure}
Let $d_j$ and $d_{j'}$ (possibly $j=j'$) be the two disks, on whose boundary arc segments $L_{arc}^j$ or $R_{arc}^j$, and $L_{arc}^{j'}$ or $R_{arc}^{j'}$ the two points $p_i$ and $p_{i'}$ lie. Based on all possible positions of $d_j$ and $d_{j'}$ for any pair $p_i$ and $p_{i'}$, we have the following cases.

\noindent\textbf{Case 2:} Assume that $p_i$ lies on $R_{arc}^j$ of $d_j$ such that its position determines the radius of the disk (there is no point $p_{i'}$ that influences the radius) (see Figure \ref{figl2c1}). For any $j=1,2,\dots,k$ in the optimal packing, $p_i$ can lie on $R_{arc}^j$ of $d_j$. For each such point $p_i\in P^{inf}$ the candidate radius $r_{\textsc{can}}$ can be calculated and stored in ${\cal L}$. Hence the number of candidate radii in this case is $O(nk)$. The candidate radius $r_{\textsc{can}}$ with respect to every point can be calculated using the below equation:
\begin{equation*}
\begin{split}
(j-1)2r_{\textsc{can}}-r_{\textsc{can}}+2r_{\textsc{can}} &= x(p_i)+\Big (r_{\textsc{can}}-\sqrt{r^2_{can}-(y(p_i)-y(q))^2}\Big)-x(p)\\
((j-1)2-1+2)r_{\textsc{can}}-r_{\textsc{can}} &= x(p_i)-\sqrt{r^2_{can}-(y(p_i)-y(q))^2}-x(p)\\
2(j-1)r_{\textsc{can}} &=x(p_i)-\sqrt{r^2_{can}-(y(p_i)-y(q))^2}-x(p)
\end{split}
\end{equation*}

The candidate radius $r_{\textsc{can}}\in {\cal L}_{\ell_{2}}$ is calculated from the above equation as we know the value of every term of the equation except $r_{\textsc{can}}$. The mirror case of this where the point $p_{i'}$ is the right most point and lies on $L_{arc}^{j'}$ of $d_{j'}$ can be handled similarly. In this case, $r_{\textsc{can}}$ can be calculated using the below equation.
\begin{equation*}
(2(k-j')+1)r_{\textsc{can}}= ||pq||-x(p_{i'})+x(p)+r_{\textsc{can}}-\sqrt{r_{\textsc{can}}^2-(y(p_{i'})-y(q))^2}
\end{equation*}
\begin{figure}[!htb]
\begin{centering}
	     \includegraphics[width=.8\linewidth]{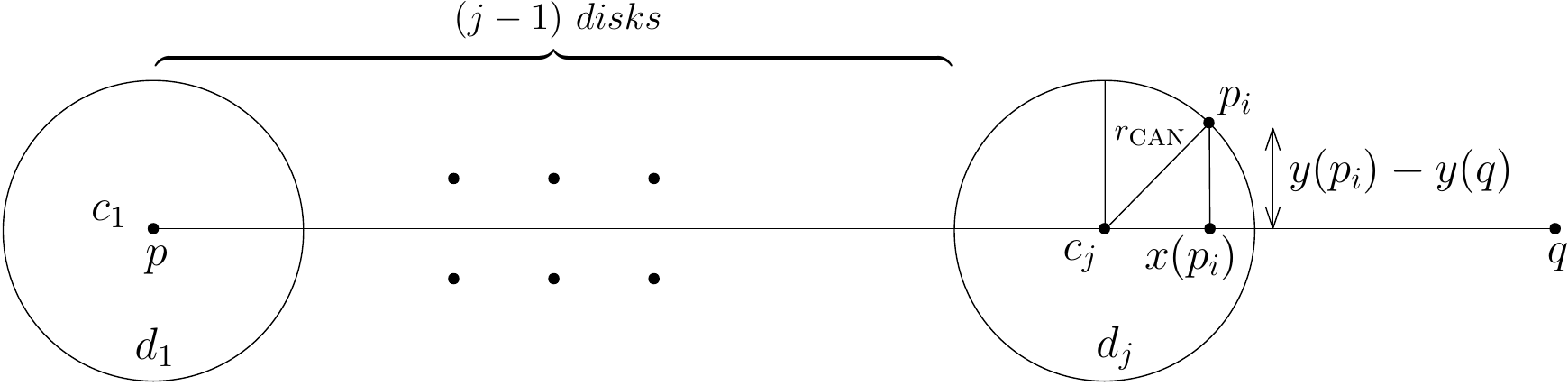}
     \caption{Illustration of case 2}\label{figl2c1}
     \end{centering} 
\end{figure}

\noindent\textbf{Case 3:} Let both the points $p_i,p_{i'} \in P$ be determining the radius of the disks in the optimal packing. Then the point $p_i$ can lie on $L_{arc}^j$ or $R_{arc}^j$ of $d_j$, where $j=1,2,\dots,k$ (see Fig \ref{figl2c2c}). Similarly, the point $p_{i'}$ can also lie on $L_{arc}^{j'}$ or $R_{arc}^{j'}$ of $d_{j'}$, where $j'=j, j+1,j+2,\dots,k$. Here, the disks centered between $p_i$ and $p_{i'}$ are compactly packed and determining the optimal radius along with the positions of $p_{i}$ and $p_{i'}$. It is easy to observe that there are a constant number $c$ of possible positions for the two disks $d_j$ and $d_{j'}$ in an optimal solution such that the two points $ p_i,p_{i'}$ lie on their boundaries and does not let their radius $r_{\textsc{can}}$ to increase by repositioning all the $k$ disks.
Therefore, the number of all candidate radii $r_{\textsc{can}}$ is ${n\choose 2}\sum\limits_{j=1}^k\sum\limits_{j'=j}^kc=k^2c{n\choose 2}=O(cn^2k^2)$. The candidate radii values $r_{\textsc{can}} \in {\cal L}$ can be computed from the following equation:
\begin{equation*}
 x(p_{i'})-x(p_{i})= 2(j'-j)r_{\textsc{can}} \pm_{(1)} \sqrt{r^2_{\textsc{can}}-(y(p_i)-y(q))^2} \pm_{(2)} \sqrt{r^2_{\textsc{can}}-(y(p_{i'})-y(q))^2}
\end{equation*} 
where in $\pm_{(1)}$, $+$ indicate $p_{i}$ lies on $L_{arc}^j$ of $d_j$ and $-$ indicate $p_{i}$ lies on $R_{arc}^j$ of $d_j$, similarly, in $\pm_{(2)}$, $+$ indicate $p_{i'}$ lies on $R_{arc}^{j'}$ of $d_{j'}$ and $-$ indicate $p_{i'}$ lies on $L_{arc}^{j'}$ of $d_{j'}$.

\begin{figure}[!htb]   
   \begin{centering}
       \includegraphics[width=.8 \linewidth]{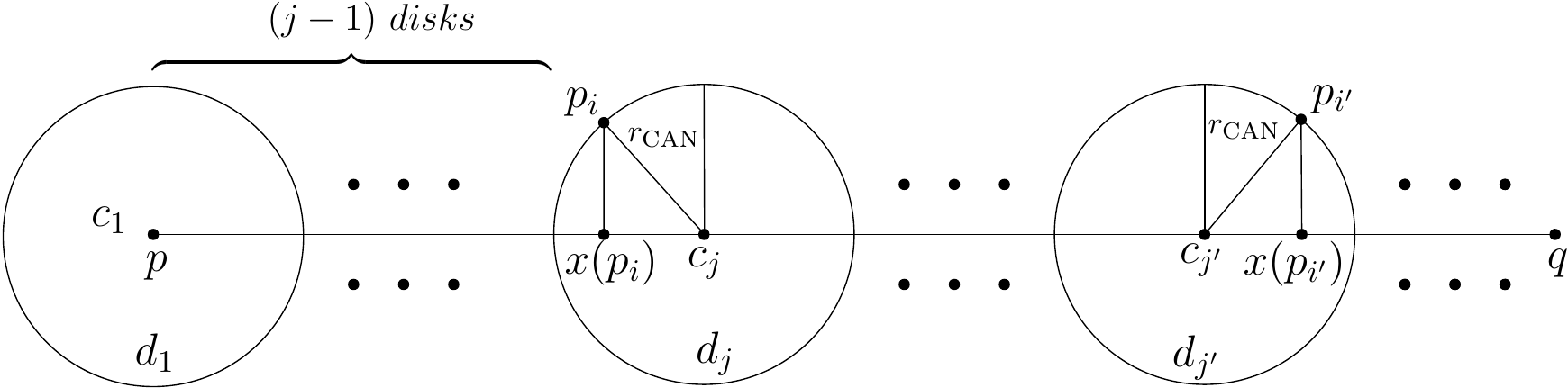}
     \caption{Two points $p_i$ and $p_{i'}$ determining the radius $r_{\textsc{can}}$}\label{figl2c2c}
\end{centering} 
\end{figure}
\noindent\textbf{Case 4:} Assume the point $p_i\in P$ lies on $R_{arc}^j$ of some disk $d_j$ and there are also other (optional) points $p_{i'}, p_{i''}\in P$ lying on boundary arcs of the disks on the right of $d_j$. Based on all posible positions of $p_i$, $p_{i'}$ and $p_{i''}$ we can observe that at most two of these points will determine $r_{\textsc{can}}$ and each of these cases corresponds to one of the above cases.
\begin{itemize}
\item[\textbf{(i)}] When $p_{i'}\in P$ lies on $L_{arc}^{j+1}$ of the disk $d_{j+1}$ and there is empty space between $d_{j}$ and $d_{j+1}$ (see Figure \ref{figl2c2a}). This corresponds to \textbf{case 2} above as either $j-1$ disks on the left of $p_i$ or $k-j-1$ disks on the right of $p_{i'}$ are compactly packed with their radius increased to the maximum possible value. 

\item[\textbf{(ii)}] When $p_{i'}\in P$ lies on $L_{arc}^{j'}$ of the disk $d_{j'}$ and $p_{i''}\in P$ lies on the disk $d_{j''}$, where $j'-j>1$ and $j''-j'>1$ (see Fig. \ref{figl2c2b6}). Then, the disks in the optimal packing can be partitioned into four subsequences of consecutive disks $d_1, d_2, \ldots, d_j$, $d_j, d_{j+1}, \ldots, d_{j'}$, $d_{j'}, d_{j'+1}, \ldots, d_{j''}$, and $d_{j''}, d_{j''+1}, \ldots, d_k$. In at least one of these, the disks must be compactly packed with their radius increased to the maximum possible value, otherwise it would contradict that we have the optimal packing. Hence, one of the above cases applies and at most two points among $p_i, p_{i'}, p_{i''}$ determine the value of $r_{\textsc{can}}$.
\end{itemize}
\begin{figure}[!htb]   
    \begin{centering}
     \includegraphics[width=.8\linewidth]{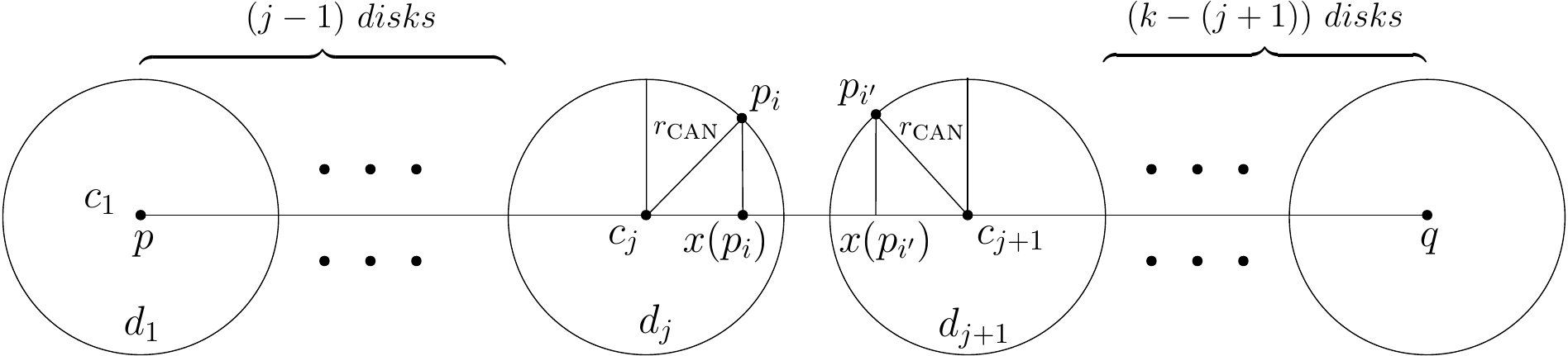}
     \caption{Illustration of case 4(i)}\label{figl2c2a}
\end{centering} 
\end{figure}
\begin{figure}[!htb]   
    \begin{centering}
     \includegraphics[width=.8\linewidth]{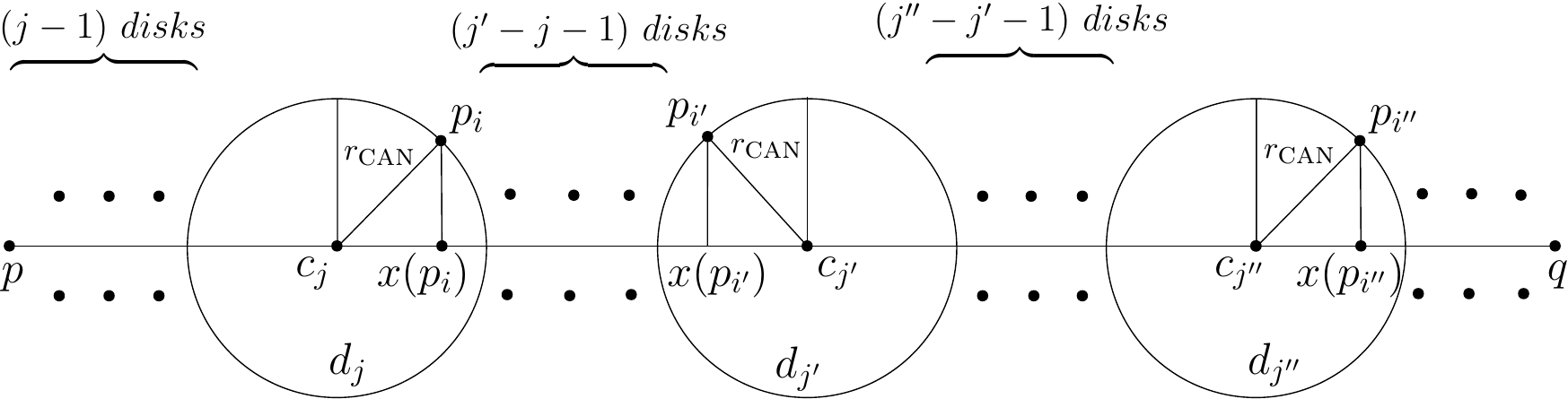}
     \caption{Illustration of case 4(ii)}\label{figl2c2b6}
\end{centering} 
\end{figure}
Thus the lemma follows.
\end{proof}
\begin{theorem}\label{th3}
For a given line segment $\overline{pq}$ and an ordered set $P$ of $n$ points in the Euclidean plane, we can solve the \textsc{COFL} problem optimally in $O((nk)^2\log{(nk)}+(n+k)\log{(nk)})$ time.
\end{theorem}
\begin{proof}
The constant $c=10$ because of the following: (i) no point lies on the boundary of any disk in optimal packing, (ii) $p_i$ lies on $L_{arc}^{j}$ or $R_{arc}^{j}$, (iii) both $p_i, p_{i'}$ lie on $L_{arc}^{j}$ or $R_{arc}^{j}$, (iv) $p_i$ lies on $L_{arc}^{j}$ and $p_{i'}$ lies on $R_{arc}^{j}$, (v) $p_i$ lies on $R_{arc}^{j}$ and $p_{i'}$ lies on $R_{arc}^{j'}$ or $p_i$ lies on $R_{arc}^{j}$ and $p_{i'}$ lies on $L_{arc}^{j'}$ or $p_i$ lies on $L_{arc}^{j}$ and $p_{i'}$ lies on $R_{arc}^{j'}$ or $p_i$ lies on $L_{arc}^{j}$ and $p_{i'}$ lies on $L_{arc}^{j'}$. Observe that (i) corresponds to {\bf case-1} and contributes 1, (ii) corresponds to {\bf case-2} and contributes 2, (iii) contributes 2, (iv) contributes 1, (v) contributes 4, hence $c=1+2+2+1+4=10$. Therefore, $|{\cal L }|=O((nk)^2)$ and we find $r_{max}\in {\cal L}$ by a binary search on sorted ${\cal L}$. Sorting ${\cal L}$ takes $O((nk)^2\log{(nk)})$ and then invoking the decision algorithm $O(\log(nk))$ times results in overall time of $O((nk)^2\log{(nk)}+(n+k)\log{(nk)})$.\end{proof}

\subsection{Algorithm based on parametric search}
The {\tt FPTAS} given in \cite{Sing2021} for the \textsc{COFL} problem is weakly polynomial because the number of calls to the decision algorithm is proportional to the number of bits of precision of accuracy for the desired approximation. Here, we discuss an algorithm based on (slower version of) Megiddo's parametric search \cite{MG1983} so that its running time is independent of any numerical precision parameter ($\frac{1}{\epsilon}$). As we know, in a solution based on parametric search, there is a test algorithm, and a decision algorithm wherein the test algorithm is typically a step-by-step simulation of the decision algorithm. We now will describe how to simulate the steps of Algorithm \ref{alg1} at the unknown maximum $L^*(=r_{max})$. 
		
Consider a point $p_i\in  P$ which is having two center points $c_{i,1}$ and $c_{i,2}$ on the segment $\overline{pq}$ (see Fig. \ref{figp}). Let the coordinates of these points be $p_i=(x(p_i),y(p_i))$, $c_{i,1}=(x_{i,1},y_{i,1})$ and $c_{i,2}=(x_{i,2},y_{i,2})$. Clearly, these points and $L$ satisfy the equations:
\begin{equation}\label{eq:1}
L^2=(x(p_i)-x_{i,1})^2+(y(p_i)-y_{i,1})^2 
\implies L^2=(x(p_i)-x_{i,1})^2+(y(p_i)-y(q))^2
\end{equation}
\begin{equation} \label{eq:2}
L^2=(x(p_i)-x_{i,2})^2+(y(p_i)-y_{i,2})^2 
\implies L^2=(x(p_i)-x_{i,2})^2+(y(p_i)-y(q))^2
\end{equation}
\hspace{.5cm} where $y_{i,1}= y_{i,2}= y(p)=y(q)$ as both the points $c_{i,1}$ and $c_{i,2}$ are located on $\overline{pq}$.

Since we know the coordinate values $y(q)$, $x(p_i)$ and $y(p_i)$, the values of $x_{i,1}$, $x_{i,2}$ are given as follows:
\begin{equation*}
\hspace{1cm}
\begin{split}
&(x(p_i)-x_{i,1})^2=L^2-(y(p_i)-y(q))^2 \implies x_{i,1}=x(p_i)-(L^2-(y(p_i)-y(q))^2)^{\frac{1}{2}},\\
&(x(p_i)-x_{i,2})^2=L^2-(y(p_i)-y(q))^2 \implies x_{i,2}=x(p_i)-(L^2-(y(p_i)-y(q))^2)^{\frac{1}{2}} \\
&\text{respectively},\ \textit{if} \  \  y(p_i)\leq y(q)+L  \ \ \  \ \ \ \ \ \ \ \ \ \ \ \ \  \  \ \ \ \ \ \ \ \ \ \ \ \ \ \ \ \ \ \  \ \ \ \ \ \ \ \ \ \ \ \ \ \ \ \ \  \ \ \ \ \ \ \ \ \ \  (3)
\end{split}
\end{equation*}

Now, consider the end points of the $i$th complemented interval as $[r_{i-1}l_i]$ (at line 5 of Algorithm \ref{alg1}). Let the coordinates $r_{i-1}=(x_{i-1,2}, y_{i-1,2})$, and $l_i=(x_{i,1},y_{i,1})$, then
\begin{equation*}
 \hspace*{-2.5cm}
\begin{split}
\gamma&=\frac{||r_{i-1}l_i||}{2L}=\frac{(x_{i,1}-x_{i-1,2})}{2L}, \\ \gamma&=\frac{x(p_i)-(L^2-(y(p_i)-y(q))^2)^{\frac{1}{2}}-(x(p_{i-1})-(L^2-(y(p_{i-1})-y(q))^2)^{\frac{1}{2}})}{2L}.
\end{split}
\end{equation*}

In the $for$-loop, at line 4 of Algorithm \ref{alg1}, we know the values of $j$ and $k$. With these values known we need to perform the comparison $j+\gamma+1-k\leq 0$. Note that this is a branching point depending on a comparison which involves a polynomial in $L$, $poly(L)$.
\begin{equation*}
\begin{split}
&j+\frac{x(p_i)-(L^2-(y(p_i)-y(q))^2)^{\frac{1}{2}}-(x(p_{i-1})-(L^2-(y(p_{i-1})-y(q))^2)^{\frac{1}{2}})}{2L}+1-k\leq 0, \\
& 2(j+1-k)L+(x(p_i)-x(p_{i-1}))-(L^2-(y(p_i)-y(q))^2)^{\frac{1}{2}}+(L^2-(y(p_{i-1})-y(q))^2)^{\frac{1}{2}}\leq 0, \\
& 2(j+1-k)L+(x(p_i)-x(p_{i-1}))\leq (L^2-(y(p_i)-y(q))^2)^{\frac{1}{2}}-(L^2-(y(p_{i-1})-y(q))^2)^{\frac{1}{2}}. \\
\end{split}
\end{equation*}
On simplification the above inequality becomes
\begin{equation*}
\begin{split}
& 2((j+1-k)^2-1)L^2+4(j+1-k)(x(p_i)-x(p_{i-1}))L +(x(p_i)-x(p_{i-1}))^2\\
& \ \ \ \ \ \ \ \ \ \ \ \ \ \ \ \ \ \ \ \ \ +2\sqrt{(L^2-(y(p_i)-y(q))^2)(L^2-(y(p_{i-1})-y(q))^2)} \leq 0 \ \ \ \ \ \ \ \ \ \ \ \ \ \ \ \ \ \ \ \ \ \ \ \ \ \ \ \ \ \ \ \ \ \ (4)
\end{split}
\end{equation*}
that involves a degree-two polynomial, $poly(L)$: $AL^2+BL+C+2((L^2-D)(L^2-E))^{\frac{1}{2}}=0$, where the coefficients $A$, $B$, $C$, $D$ and $E$ depend on the values known at the point of execution of the corresponding comparison step. To simulate comparison steps in the \textit{for} loop at line 6, we compute all the roots of the associated polynomial $poly(L)$ and invoke Algorithm \ref{alg1} with the value of $L$ equal to each of these roots. This yields an interval between two roots or a root and $0$ or $|pq|/(2(k-1))$ that contains $L^*$. This enables us to determine the sign of $poly(L^*)$ and proceed with the generic execution of the next step of the algorithm. Essentially, each time Algorithm \ref{alg1} returns \textsc{yes} for the guessed value of $L$, the region ${\cal R}$ (which depends on the right end of interval containing $L^*$) is shrinking, and hence the number of complemented intervals in $I^\mathsf{c}$ on which we have to run Algorithm \ref{alg1} is reducing, until the shrunken ${\cal R}$ corresponds to $L^*$. Finally, after completing the simulation of the \textit{for} loop, if $j=k$, then we return $L^*$. To construct the set $D=\{d_1,d_2,\dots,d_k\}$, we run Algorithm \ref{alg1} with the computed $L^*$ one more time. Therefore, we have the following theorem.

\begin{theorem}\label{thm4}
We have an algorithm to solve the \textsc{COFL} problem in $O((n+k)^2)$ using parametric search technique.
\end{theorem}
\begin{proof}
From \cite{Sing2021} we know that Algorithm \ref{alg1} runs in $O(n+k)$ time. In the worst case, for each step of Algorithm \ref{alg1} we obtain two different values of $L$ and invoke Algorithm \ref{alg1} on each of them as the candidate radius. Let $L_1,L_2, \dots, L_t$ be those different values of $L$ across the entire simulation, where $t=O(n+k)$. Then, initially $r_{max} \in [0,\frac{||pq||}{2(k-1)}]$. After the entire simulation is completed, clearly $r_{max}=L^*=\max\{L_u \mid \textsc{Dcofl}(P, k, L_u)=\textsc{yes}, \ u=1,2,\dots,t\}$. Since degree of the polynomial in inequality (4) is at most 2 and the decision algorithm \textsc{Dcofl}$(P, k, L_u)$ is monotone for any $L_u\in {\rm I\!R}^+\cup\{0\}$, the entire setup fits in the framework of parametric search. Hence, the correctness of the algorithm follows and the overall time for the simulation is $O((n+k)^2)$.
\end{proof}

\subsection{Improved algorithm for $k=2$}

%

Here, we show that the decision problem \textsc{Dcofl}$(P, 2, L)$ can be solved in $O(\log n)$ parallel time using $n$ processors. Let $Q(L,P)$ be the complement of the union of $n$ open disk of radius $L$ centered at the demand points in $P$. Let $S(\overline{pq}, L)$ be the intersection of $Q(L,P)$ and $\overline{pq}$, which is the collection $I^\mathsf{c}$ of $O(n)$ disjoint feasible intervals $[r_{i-1}l_i]\subset \overline{pq}$, $i=1,2,\dots, m=O(n)$, where the coordinates $r_{i-1}=(x_{i-1,2}, y_{i-1,2})$, and $l_i=(x_{i,1},y_{i,1})$. Let an infeasible interval $[x_{i,1},x_{i,2}]$ be an interval on $\overline{pq}$, which is not feasible for centering a facility in that (see Figure \ref{figfp}). When we have the infeasible intervals computed, implicitly we also have computed the feasible intervals in $I^\mathsf{c}$. A parallel algorithm for computing these intervals is as follows: (1) We assign one demand point from $P$ to each of the $n$ processors. (2) Let each compute the corresponding infeasible interval $[x_{i,1},x_{i,2}]=d_i\cap \overline{pq}$, where $d_i$ is the open disk of radius $L$ centered at $p_i\in P$ for $i=1,2,\dots, n$ (see Figure \ref{figfp}). (3) Then, the merging of consecutive overlapping infeasible intervals into one bigger infeasible interval is performed by processors as follows: If an interval is not overlapping with its adjacent intervals then this interval is maintained on the same processor and for each pair of consecutive overlapping intervals, the processor of the first interval merges them into one and the second processor sits idle. In this way, for a sequence of consecutive overlapping intervals, alternating processors will perform the merging. This process will be repeated until there are only isolated intervals. Since in every step we merge a pair of consecutive overlapping intervals there will be at most $\log{n}+1$ steps in total and each step will take $O(1)$ parallel time. Also, the serial time to construct $S(\overline{pq}, L)$ is only $O(n)$.

\begin{figure}[!htb]

\begin{centering}
     \includegraphics[scale=0.7]{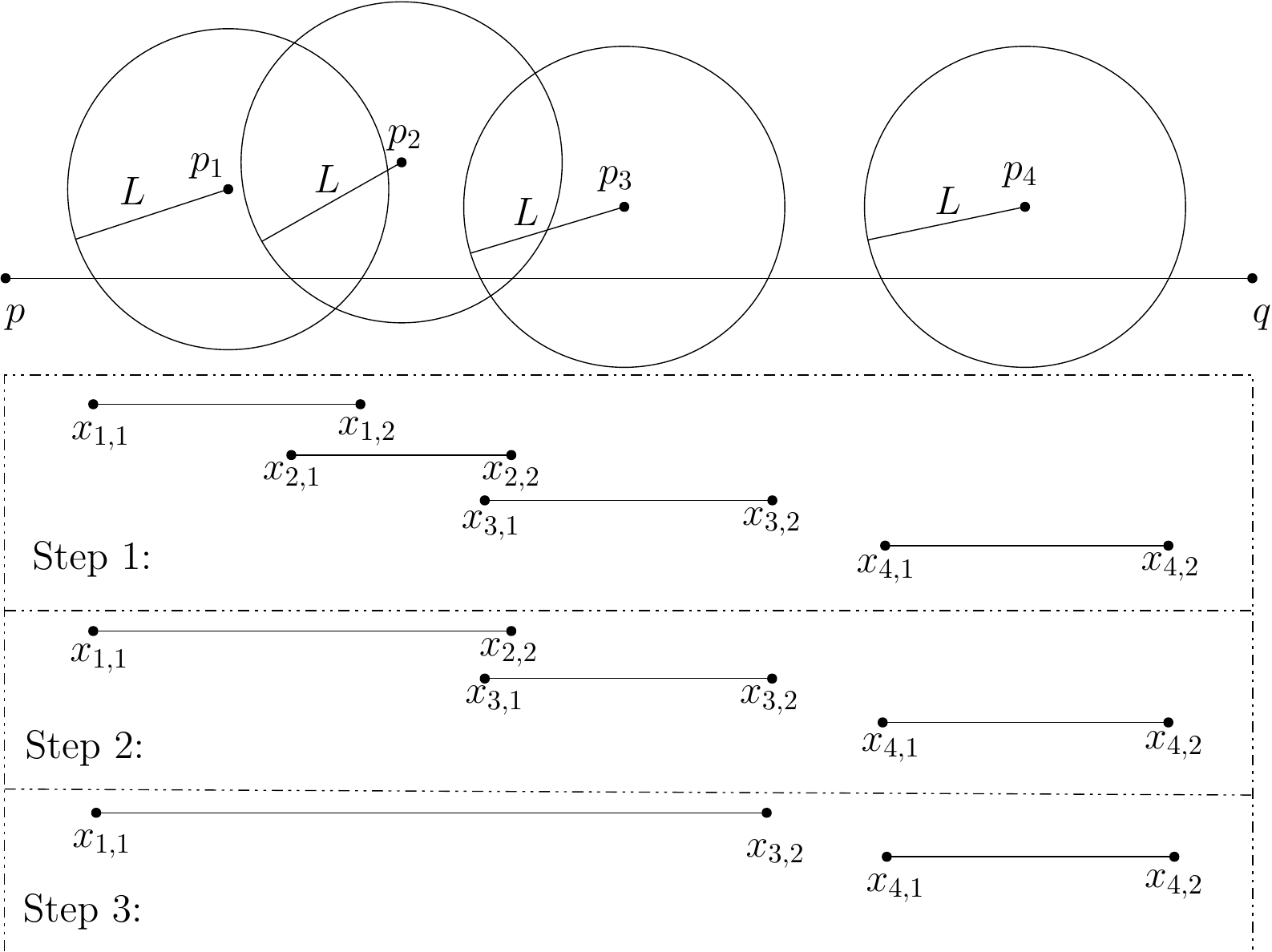}
     \caption{Merging of intervals in parallel algorithm} \label{figfp}
     \end{centering} 
\end{figure}

Now choosing two farthest points on the intervals in $S(\overline{pq}, L)$ is easy, just pick the left end point of the left most interval and right end point of the right most interval in $I^\mathsf{c}$, in $O(1)$ time, and place the two facilities centered at these points.

Therefore, overall time of the parametric algorithm to solve \textsc{COFL} for $k=2$ is $O(T_p\cdot n\cdot\log{n}+T_p\cdot T_s \cdot\log{n})=O(\log{n}\cdot n\cdot\log{n}+\log{n}\cdot n\cdot \log{n})=O(n\log^2{n}) \text{time}$, where $T_p$ denotes a parallel time and $T_s$ denotes a serial time for solving \textsc{Dcofl}$(P, 2, L)$ with $n$ processors.

This is an improvement over the earlier {\tt FPTAS} as well as the two proposed exact algorithms here for $k=2$.

\begin{remark}
The \textsc{COFL} under the Euclidean norm for $k=2$ can be solved in $O(n\log^2{n})$ time using the improved  parametric technique.
\end{remark}
%
%
%
%

\section{Circular COFL problem}
 In this section, we define a variant of the \textsc{COFL} problem in which the centers of the disks are restricted to lie on the boundary arc of a predetermined circle. The motivation for this variant of the problem is as follows. Consider a city, and now we need to place $k$ number of facilities such as dumping yards or fuel stations around the city such that the minimum distance between these facilities and between the facilities and demand points has to be maximized. It is also required to place these $k$ facilities as far apart from each other as possible along the boundary arc $\partial {\cal C}$ to avoid self competition among fuel stations and to avoid one of two or more close-by dumping yards getting overloaded quickly. 
 
The \textit{circular constrained obnoxious facility location problem} (\textsc{CCOFL}) problem is defined as follows: Given a set $P=\{p_1,p_2,p_3, \dots ,p_n \}$ of $n$ demand points in the plane, a predetermined circle ${\cal C}$ with radius $r_c$ and a positive integer $k$, locate $k$ facility sites on the boundary arc $\partial {\cal C}$ of ${\cal C}$ such that each demand point in $P$ is farthest from its closest facility site (i.e., in terms of their Euclidean distance) and the facility sites are placed farthest from each other along $\partial {\cal C}$ (i.e., in terms of arc length). Observe here that two disks representing two consecutive facility sites placed on $\partial {\cal C}$ may overlap strictly inside ${\cal C}$.  

The decision version of this problem \textsc{Dccofl}$(P, k, L)$ can be solved by using a similar method that was used to solve the \textsc{Dcofl}$(P, k, L)$ problem under the assumption that $r_c>>L$. When $r_c<L$, it is trivial that only one disk with radius $L$ will be packed on $\partial {\cal C}$ as the center of ${\cal C}$ will lie inside that packed disk. Now, to solve the decision version of \textsc{CCOFL} problem, we consider two circles ${\cal C}_1$ and ${\cal C}_2$ which are concentric with ${\cal C}$ and whose radii are $r_c-L$ and $r_c+L$, where $L$ is a real number (see Figure \ref{figc1}). We can observe that the points lying inside ${\cal C}_1$ and outside of ${\cal C}_2$ will not influence the packing of the disks. Similar to that of \textsc{Dcofl}$(P, k, L)$ problem here, we can obtain at least one center point and at most two center points on the boundary of ${\cal C}$ which are at a distance of $L$ from each point lying outside of ${\cal C}_1$ and inside of ${\cal C}_2$ (see $p_i$, $p_j$ and $p_k$ in Figure \ref{figc1}).

\begin{figure}[!htb]

\begin{centering}
     \includegraphics[scale=0.7]{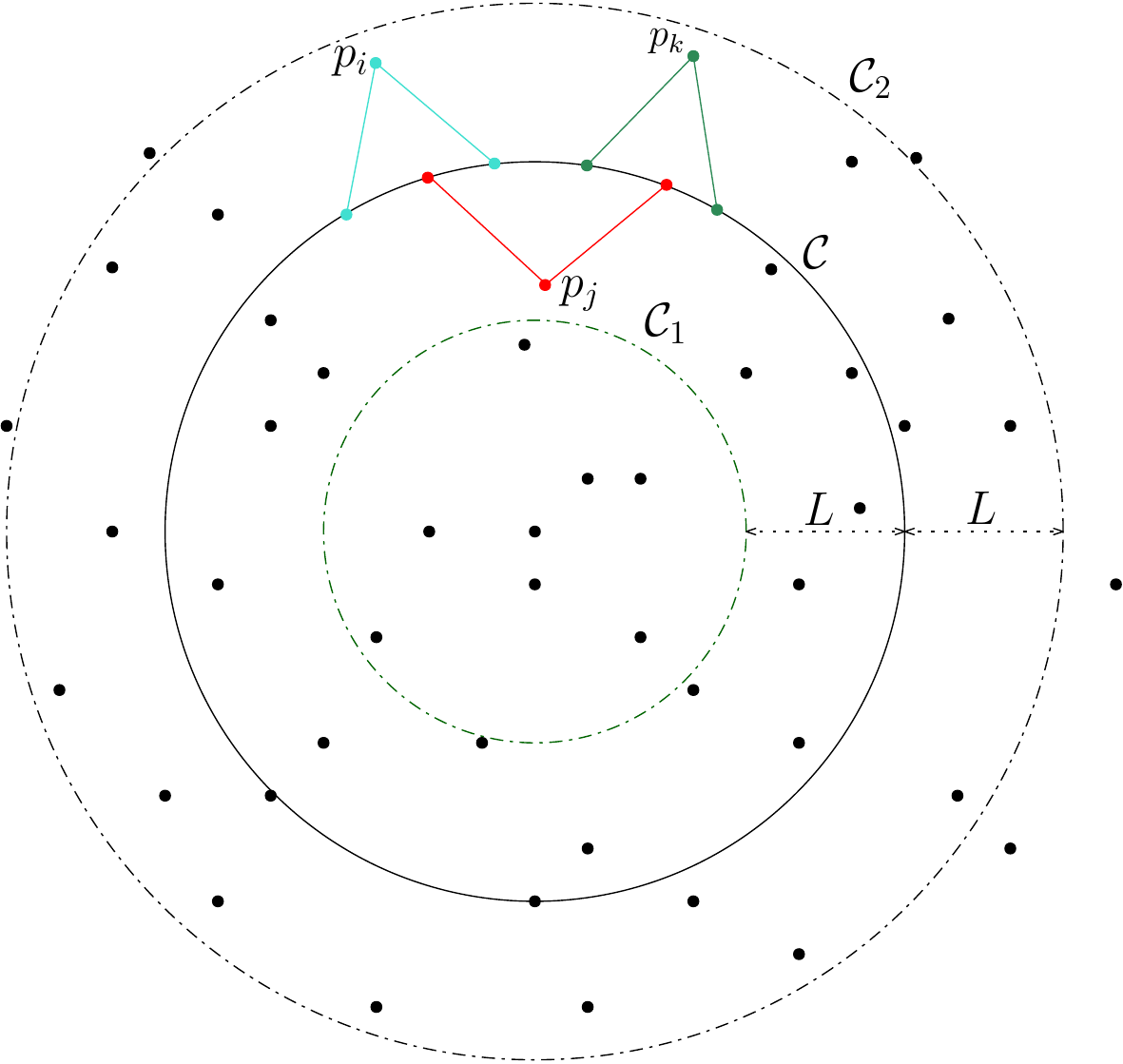}
     \caption{Circles ${\cal C}_1$ and ${\cal C}_2$ which are at distance $L$ from ${\cal C}$} \label{figc1}
     \end{centering} 
\end{figure}

Let $c_{i,1}$ and $c_{i,2}$ be the center-points corresponding to $p_i$, then none of the $k$ disks in an optimal solution to \textsc{Dcofl}$(P, k, L)$ will have their center points lying on the open arc interval $(c_{i,1}, c_{i,2})$ of the boundary of ${\cal C}$. 

Now, let $(c_{j,1}, c_{j,2})$ and $(c_{k,1}, c_{k,2})$ be the center points corresponding to $p_j$ and $p_k$ respectively.
In Figure \ref{figc1}, we can observe that the intervals $[c_{i,1}, c_{i,2}]$, $[c_{j,1}, c_{j,2}]$ and $[c_{j,1}, c_{j,2}]$ formed by $\tikzarc{{c_{i,1}c_{i,2}}}$, $\tikzarc{{c_{j,1}c_{j,2}}}$ and $\tikzarc{{c_{k,1}c_{k,2}}}$ are overlapping. Hence, none of the $k$ disks in the optimal solution will have their centers lying on the interval $([c_{i,1}, c_{i,2}]\cup [c_{j,1}, c_{j,2}]\cup [c_{k,1}, c_{k,2}])\setminus\{c_{i,1},c_{k,2}\}$, excluding the end-points of the union of the two intervals.

Without loss of generality, let $\{p_1, p_2, \ldots, p_m\}$  be the points of $P$ lying strictly outside of ${\cal C}_1$ and inside of ${\cal C}_2$, ordered clockwise on the boundary of ${\cal C}$, where $m\leq n$. We know that for every point $p_i$ lying strictly outside of ${\cal C}_1$ and inside of ${\cal C}_2$ there will be two center-points on the boundary of ${\cal C}$ which are at distance $L$, i.e., there is an interval $[l_i, r_i]$ for every point $p_i$, where $l_i=c_{i,1}$ and $r_i=c_{i,2}$. Merge all the overlapping intervals and then update the end-points of the new intervals on boundary of ${\cal C}$. Let $I=\{[l_1, r_1], [l_2, r_2], \ldots, [l_{m'}, r_{m'}]\}$ be the set of resulting pairwise disjoint intervals ordered clockwise, where $m'\leq m$.

Consider the complement of $I$ with respect to boundary of the circle ${\cal C}$, denoted as \[I^\mathsf{c}=\{[r_1, l_2], \ldots, [r_{m'}, l_1]\}.\]

The arc length of the complemented intervals in $I^\mathsf{c}$ can be calculated using the law of cosines formula as follows: The angle ($\theta$) subtended by arc at the center of ${\cal C}$ is 
\[\theta=\arccos{(1-\frac{d^2}{2r_c^2})} \]
where $d$ is the Euclidean distance between the end-points of the arc segment. Then, the length of the arc interval $[r_il_{i+1}]$ is $||\tikzarc{r_il_{i+1}}||=r_c\theta$, where $\theta=\arccos{(1-\frac{d^2}{2r_c^2})}$ for $i=1,2,\dots,m'$.

\begin{observation}\label{cob1}
Without loss of generality, we can assume that the first disk $d_1$ is centered at one of the end-points of arc segment in $I^c$.
\end{observation}

Since we don't know, the position of the disks for given $L$ in the optimal packing, we greedily pack disks by placing centers on $\partial {\cal C}$ with the first disk $d_1$ at every endpoint of the  arc segments in $I^c$. As there are $O(n)$ end-points in $I^c$, we have the following theorem.

\begin{lemma}\label{cl1}
Given the set $I^\mathsf{c}$ of complemented intervals and an integer $k>0$, Algorithm \ref{alg1} solves the \textsc{Dccofl}$(P, k, L)$ problem in $O(n(n+k))$ time.
\end{lemma}
\begin{proof}
Follows from Observation \ref{cob1} and the fact that there are $O(n)$ end points of intervals in $I^c$.
\end{proof}
\begin{theorem}\label{ct1}
We can get an $(1-\epsilon)$-factor approximation algorithm with $\epsilon>0$ ({\tt FPTAS}) for the {\it CCOFL} problem, that runs in $O(n(n+k)\log(\frac{||pq||}{2(k-1)\epsilon}))$ time, by employing doubling search and bisection methods.
\end{theorem}
\begin{proof}
Follows from Lemma \ref{cl1} and Theorem in \cite{Sing2021}.
\end{proof}
\begin{remark}
\cite{Sham1975} Unconstrained \textsc{COFL} problem can be solved in $O(n\log{n})$ time for $k=1$ by finding the largest empty circle in a Voronoi diagram formed by $n$ demand points as the voronoi sites.
\end{remark}

%
%
%
%

\section{Conclusion and future work}
We proposed a brute force algorithm based on binary search that solves the \textsc{COFL} problem exactly in $O((nk)^2\log{(nk)}+(n+k)\log{(nk)})$ time for both rectilinear and Euclidean cases. 
We showed that using Megiddo's parametric search technique, we can improve the running time to $O((n+k)^2)$ time. For $k=2$, we have a parametric algorithm that runs in $O(n\log^2{n})$ which would improve the parametric algorithm. It is not clear how to design a parallel algorithm for \textsc{Dcofl}$(P,k, L)$ for any $k$. We also gave a {\tt FPTAS} for \textsc{CCOFL} problem that runs in $O(n(n+k)\log(\frac{||pq||}{2(k-1)\epsilon}))$ time.

For any $k\geq 2$, unconstrained continuous \textsc{OFL} is an elusive open problem since it was first mentioned by Kartz et al. \cite{Katz2002}.  However, for rectilinear variant \textsc{UCOFL} is tractable at least from exact exponential algorithm, and for $k=1,2$, there are polynomial time algorithms.

\small
\bibliographystyle{abbrv}

\end{document}